\documentclass[12pt,leqno]{amsart}
\usepackage{amsmath,epsfig,graphicx,color, float}
\usepackage{subcaption}
\usepackage{relsize}
\usepackage{comment}
\usepackage[colorlinks, linkcolor=blue,citecolor=blue]{hyperref}
\usepackage{pdfpages}
\usepackage{dsfont}
\usepackage{graphicx}
\usepackage{mwe}
\usepackage{algorithm}
\usepackage{algorithmicx}
\usepackage{algpseudocode}
\usepackage{natbib}
\usepackage{bbm}
\usepackage[margin=1.3in]{geometry}
\numberwithin{table}{section}
\numberwithin{figure}{section}
\numberwithin{equation}{section}

\definecolor{darkblue}{rgb}{.2, 0.2,.8}
\definecolor{darkgreen}{rgb}{0,0.5,0.3}
\definecolor{darkred}{rgb}{.8, .1,.1}

\newcommand{\x}{{x}}
\newcommand{\X}{{X}}

\newcommand{\Q}{\mathbb{Q}}
\newcommand{\R}{\mathbb{R}}
\newcommand{\E}{\mathbb{E}}
\renewcommand{\P }{{\mathbb P}}
\newcommand{\M}{\mathcal{M}}

\newcommand{\GG}{\Xi}
\newcommand{\pr}{\xi}

\newtheorem{lemma}{Lemma}[section]
\newtheorem{theorem}[lemma]{Theorem}

\newtheorem{definition}[lemma]{Definition}

\newtheorem{example}[lemma]{Example}

\newtheorem{remark}{Remark}[section]

\usepackage{bm}

\allowdisplaybreaks

\begin{document}
\title[Informed censoring]{Informed censoring: the parametric combination of data and expert information
}

\author[H. Albrecher]{Hansjörg Albrecher}
\address{Faculty of Business and Economics,
University of Lausanne,
Quartier de Chambronne,
1015 Lausanne,
Switzerland}
\email{hansjoerg.albrecher@unil.ch}

\author[M. Bladt]{Martin Bladt}
\address{Department of Mathematical Sciences, 
University of Copenhagen,
Universitetsparken 5, 
DK-2100 Copenhagen, 
Denmark}
\email{martinbladt@math.ku.dk}

\begin{abstract}
The statistical censoring setup is extended to the situation when random measures can be assigned to the realization of datapoints, leading to a new way of incorporating expert information into the usual parametric estimation procedures. 
The asymptotic theory is provided for the resulting estimators, and some special cases of practical relevance are studied in more detail. 
Although the {proposed framework} mathematically generalizes censoring and coarsening at random, and borrows techniques from M-estimation theory, it provides a novel and transparent methodology which enjoys significant practical applicability in situations where expert information is present. 
The potential of the approach is illustrated by a concrete actuarial application of tail parameter estimation for a heavy-tailed MTPL dataset with limited available expert information. 
\end{abstract}
\maketitle

\section{Introduction}\label{sec1}

There are many application areas of parametric statistics where there is some uncertainty about the actual size of some (or all) datapoints. In survival analysis, such an uncertainty appears naturally for censored data, see for instance \cite{leung1997censoring} and \cite{lesaffre}. 
In lack of further information,  the realization of each datapoint may then be considered equally likely at any point within the respective censoring interval, making a uniformity assumption natural, which is at the basis of the classical censoring literature.

There may, however, also be situations in which one has additional information on how likely the realization of a datapoint within an interval is. This may, for instance, 
be the case when one has to consider physical measurement errors in connection with the involved measurement devices or potentially accumulated rounding errors in the data gathering process leading to the eventually reported value.

In actuarial applications, in order to get a better understanding of the tail risk and in view of a limited amount of available datapoints, regulators may, for instance, ask for the generation of further adverse scenarios that lead to additional, more extreme outcomes in the dataset. Or, in the long-tailed insurance business, the settlement of claims (that then serve as datapoints) may take several or even many years, and in addition to a lower bound on the final value (which is typically the already paid amount), there may be an expert opinion on the final value (the `\textit{ultimate}') in terms of statistical loss development techniques, insight into the concrete case, or simply `experience'. Often the received value is then plainly added to the dataset as an additional datapoint, although one may have a concrete idea about the uncertainty of that point (for instance, the variance around such an estimate above may decrease with the number of years of the claim in the settlement process). Figure \ref{ultimates_vs_paid} illustrates these different possibilities to incorporate the information on a datapoint, where the third option is the one to be followed upon in the present paper (the yet unknown true value is depicted in grey, and the green line depicts another fully settled (closed) claim for illustration purposes).

\begin{figure}[h]
	\centering
	\includegraphics[clip, trim=5cm 5cm 18cm 5cm,width=0.39\textwidth]{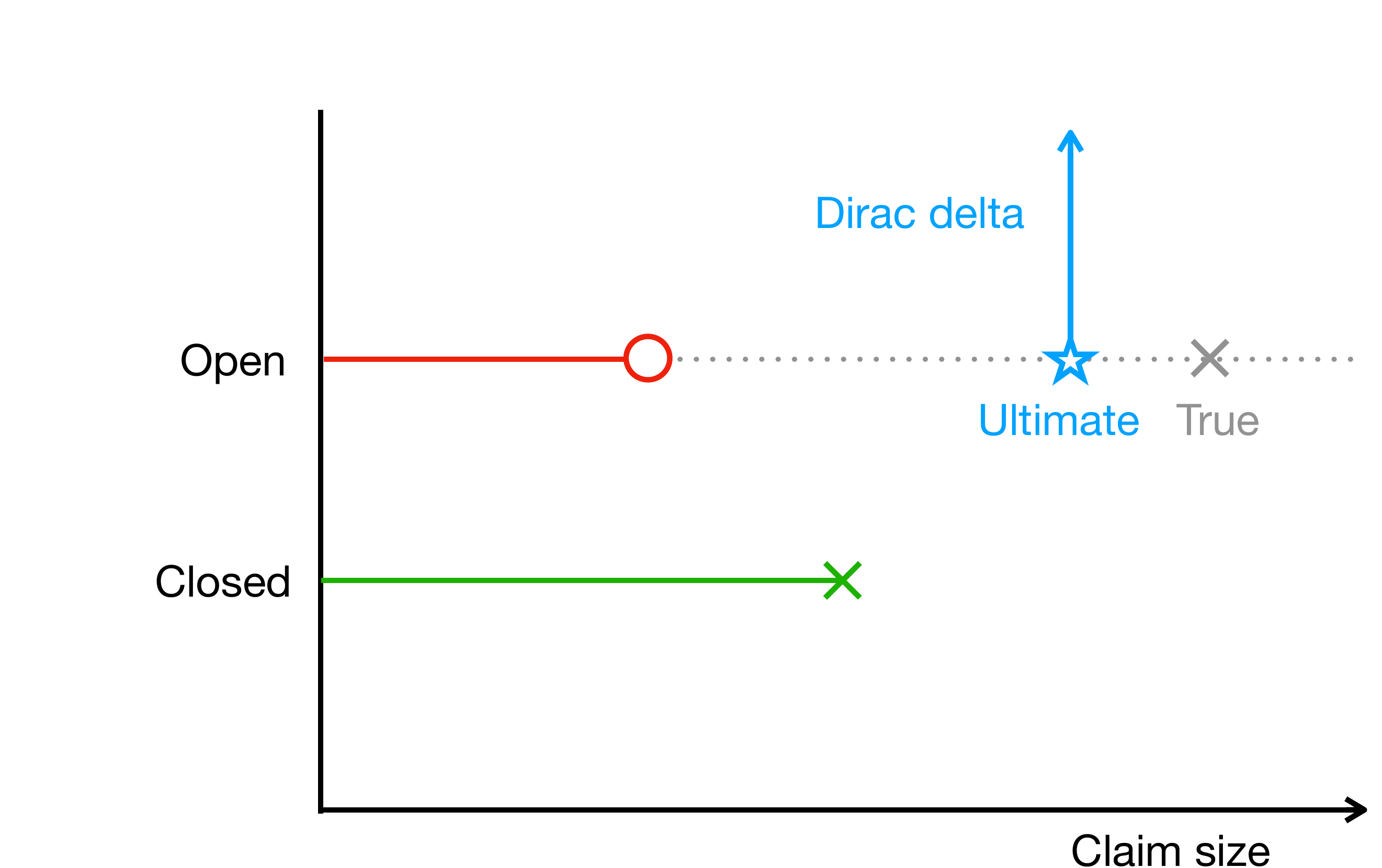}
	\includegraphics[clip, trim=15cm 5cm 18cm 5cm,width=0.29\textwidth]{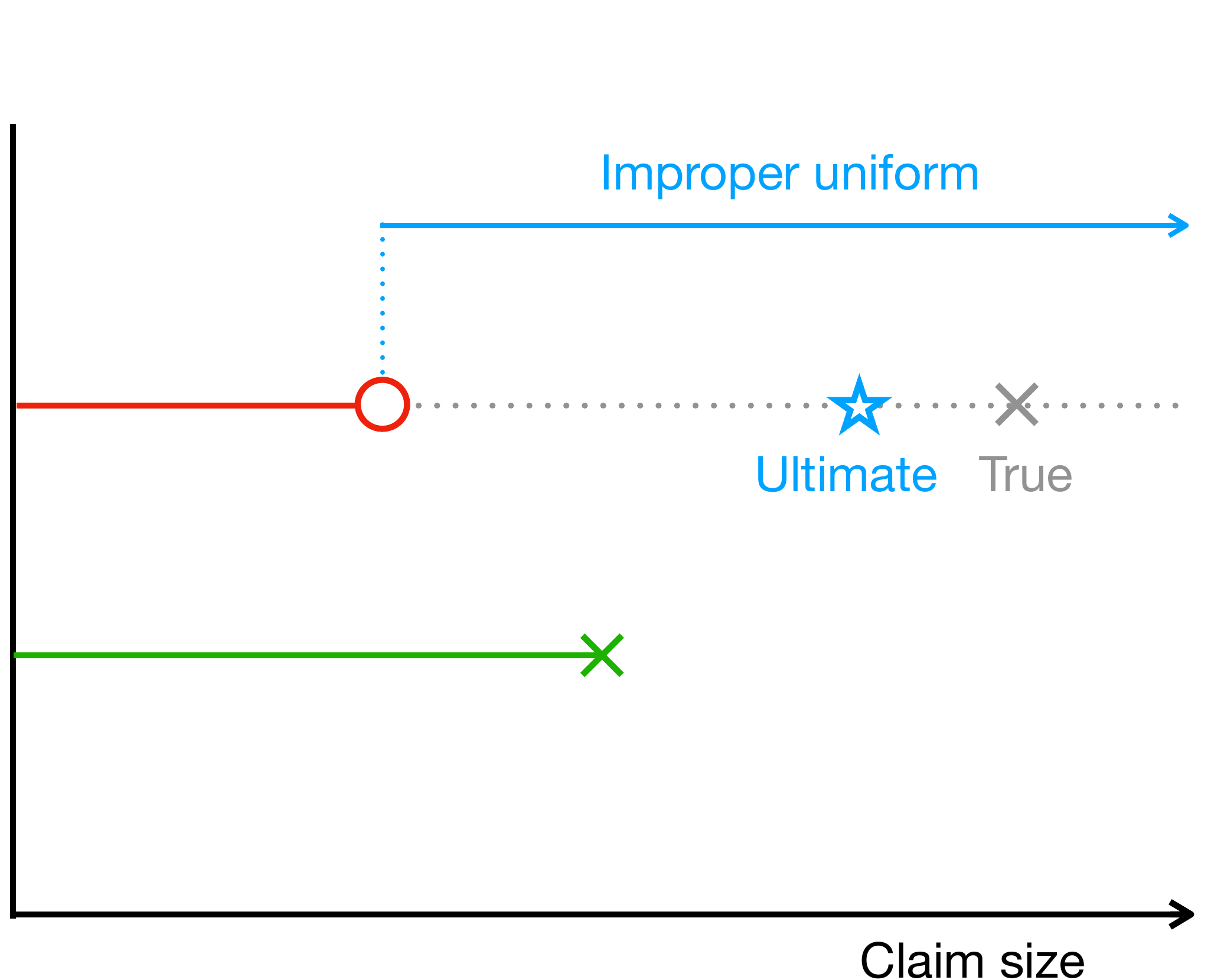}
	\includegraphics[clip, trim=15cm 5cm 18cm 5cm,width=0.29\textwidth]{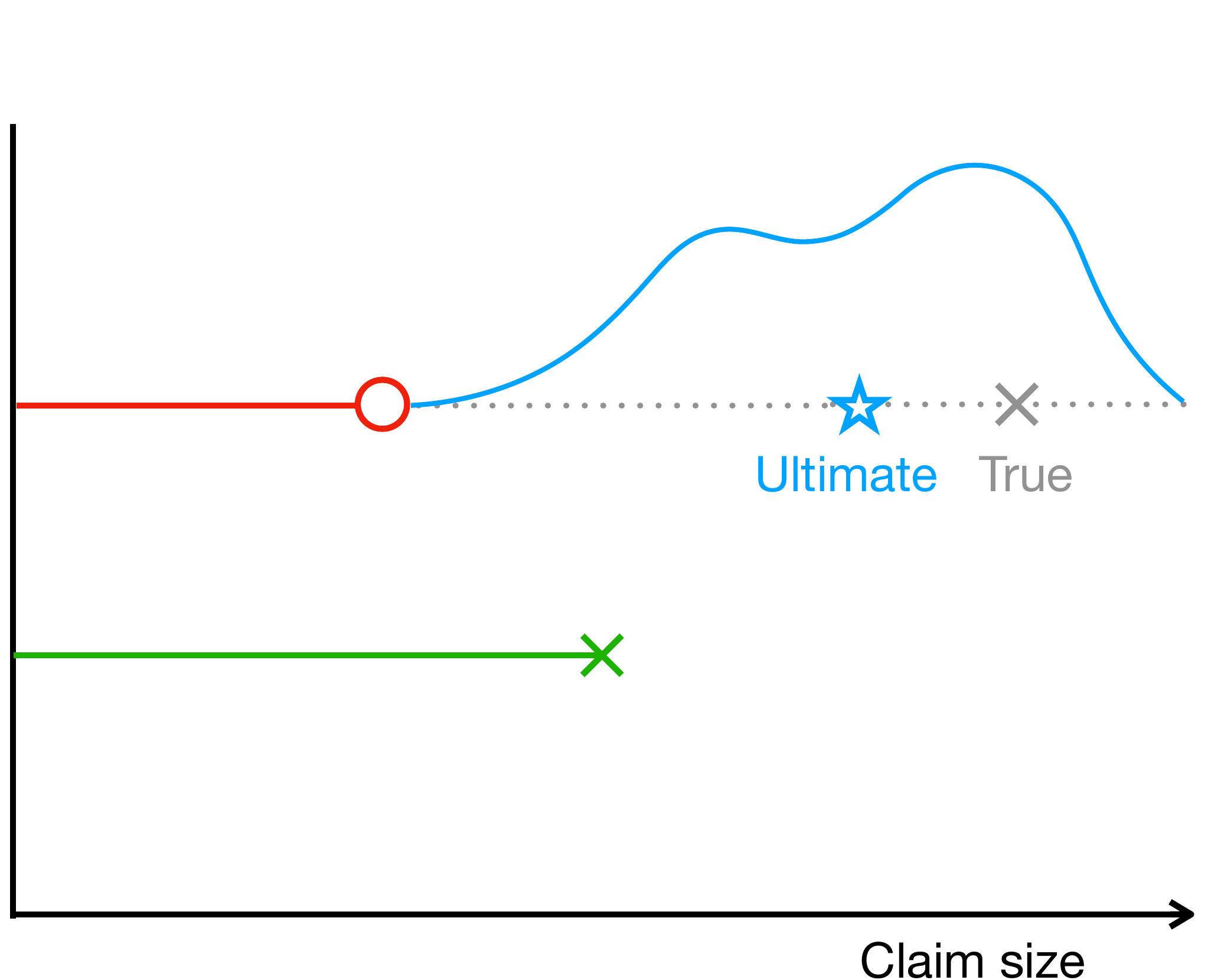}
	\caption{Three possibilities to include information on a datapoint in the actuarial dataset: Imputing the ultimate value (left), taking the already paid amount as the lower bound in an interval censoring approach (middle) or replacing the ultimate value by a distribution around its value (right).}
	\label{ultimates_vs_paid}
\end{figure}
While the above actuarial examples serve as a concrete motivation and framework for the questions addressed in this paper (and we will return to that field in the illustration section), the setup is general and has potential applications in other areas as well. Expert judgement is usually carried out by imputation or combination. The first of these approaches modifies the data before the analysis is made, while the second often modifies estimators after they have been constructed. 

In this paper we propose a formal way of incorporating uncertainty in the form of expert information or judgment which happens precisely at the estimation step. Concretely, we study how to incorporate information on the (possibly improper) `distribution' of the true value of an actual datapoint in the estimation process. This extends the classical statistical estimation literature by replacing each datapoint by a random measure. One can a priori imagine that there will be a certain trade-off between the `inaccuracy' of datapoints and the number of datapoints needed to lead to a comparable efficiency when compared to a situation with `exact' datapoints (Dirac-delta measures in the present setup). At the same time, the asymptotic theory in the sequel will show that the random measures introduce a bias in the estimation procedure so that one in general loses consistency of the resulting estimator. As we will show, this bias, however, can reduce the spread (mean square error) of the estimators considerably, so that -- like in many other fields of statistics -- such a tradeoff may be worthwhile (see e.g.\ a similar compromise in the context of Hill estimation for the extreme value index, cf.\ \cite{BGST2004}, and in terms of regularized regression in \cite{tibshirani1996regression}). This tradeoff can be particularly interesting in applications where a more accurate measurement of datapoints may be possible but costly (like in certain drug testing experiments) {or simply take too long (like in the actuarial application mentioned earlier)}. 
 
Apart from the above-mentioned practical situations in which such an approach can be useful, one may also be motivated by pure mathematical curiosity (and appetite for randomization), {namely} to what extent the corresponding statistical theory extends to the setup of random measures. It will turn out that the proof techniques in principle follow the lines of M- and Z- statistics (with the relaxation of some measurability assumptions). Yet the setting is different: when for instance comparing our setup with the one of mis-specified models in M-theory (see e.g.\ \cite[Example 5.25]{van2000asymptotic}), conceptually one may interpret that in the latter framework the model is `wiggled', introducing a systematic bias, whereas in our situation that wiggling applies to the datapoints used for estimating the parameters in a correctly specified model, {which also introduces a bias.}

We would like to note that random censoring (see e.g.\ \cite{worms2021estimation,wang2022estimation,goegebeur2023conditional}) is different in spirit to our current setup, but can easily be accommodated here as well.  Concretely, our theory contains the parametric likelihood of a randomly censored sample when the random measures have a very special form: having a constant improper density from the censoring upper limit onwards. The more general 'coarsening at random' (see \cite{heitjan1991ignorability}, and also \cite{heitjan1993ignorability} for special cases and applications to biomedical data) is also somewhat related. {\color{black}Their motivation and approach is, however, conceptually different from ours: they consider a random coarsening procedure of original datapoints (e.g.\ through rounding, heaping or grouping of datapoints) that one tries to disentangle, and while one can frame their analysis as a particular case of a random measure likelihood technique, the focus is on conditions under which such a coarsening can be ignored for the resulting estimates, exploiting the structure of related random sets. In contrast, we study the consequences of intrinsic randomness (of any kind) of the true values of datapoints (not necessarily linked to a particular causal structure behind that randomization), which is a general setting and in particular allows an interpretation and formalization of expert information that can refine the knowledge of datapoints beyond the classical coarsening or censoring mechanisms. }

{\color{black} Related contributions to the actuarial literature in terms of expert information or judgment are \cite{ Tredger2016}, \cite{combined} and \cite{Bladt2022}. The first reference aims at aiding the modeller to distinguish between low-quality and high-quality information, while the latter two focus on perturbed likelihood and Kaplan--Meier estimation, respectively. None of these contributions carries out an abstract and general treatment of expert information as a fully flexible random measure to be included in the likelihood, which is a main advantage of our proposed method.  In clinical trials, incomplete event adjudication has been treated in \cite{CookKosorok2004}, where some subjects in the dataset might have an uncertain status. In that case, consistent estimators are available when one assigns suitable weights to each unadjudicated observation, but the procedure only deals with uncertainty in the binary status variable, and once again does not allow for a full distributional imputation when more precise information is available. We would also like to mention the use of expert information in medicine in \cite{JULIAFLORES2011181} using a combination of Bayesian Networks and expert elicitation (cf.\ also \cite{quigley2021characteristics,bojke2021developing}), and in dependence modelling in \cite{arbenz2012estimating} using a Bayesian approach to incorporate expert information into copula estimation to improve precision, see also \cite{ashcroft2016expert} for a general discussion on the topic. An appropriate framework for the quantitative inclusion of expert information in models is also relevant in other fields with far-reaching impact, like the study of effects of climate change, see e.g.\ \cite{Mach2017,hurlbert2019risk}. 

}

The remainder of the paper is structured as follows. In Section \ref{sec2} we introduce the general setup and approach for the informed censoring studied in this paper. In particular, we define the appropriate generalized maximum likelihood estimators for the parameter(s) of interest. Section \ref{sec3} then provides the asymptotic theory of the resulting estimators, with the proofs delegated to the appendix. In Section \ref{sec:dens_rand}, we discuss the case of proper densities for the datapoint locations in more detail and provide two concrete examples for which the asymptotic bias and variance can be calculated explicitly. We then also suggest a way to measure the asymptotic efficiency of considering the available information relative to the situation of knowning the exact values of these datapoints. This last point leads to a possibility to quantitatively tradeoff the uncertainty of the datapoints to the associated sample size required to achieve the same efficiency in a fully observed scenario. Section \ref{sec5} then deals with the case of improper randomizations of datapoints, which plays an important role when generalizing the survival likelihood, and is a key component for the empirical analysis to follow. In Section \ref{sec6} we then apply the proposed approaches to the concrete actuarial application discussed above, highlighting the potential of the method. Section \ref{seccon} concludes.

\section{The setup}\label{sec2}


In the sequel, unless stated otherwise, $$f:\GG\times\mathbb{R}^d\to\mathbb{R}_+,\quad \GG\subset \mathbb{R}^p,$$ will be a parametric multivariate density function, that is, $\int_{\R^d}f(c,\x)d\x=1$, for any parameter $c\in\GG$. We will often write $f_{c}(\x)$ for  $f(c,\x)$, or simply $f(\x)$ when the role of $c$ is not important. When $d=1$, that is, the univariate case, the associated cumulative distribution function will be denoted by $F(x)=\int_{-\infty}^xf(y)dy$. The parametric estimation of $f$ is assumed to be of primary concern.

Let the support of $f$ be given by $\mathcal{S}=\mbox{supp}(f)=\bigcap\{\overline{A}\in \mathcal{B}(\R^d)\mid \int_{A^c}f(\x) d\x=0\}$. All random elements in the sequel are defined on a common probability space $(\Omega,\mathcal{F},\mathbb{P})$, which by extension can be done so without loss of generality. 

We now introduce the concept of informed censoring in an abstract manner. Examples in subsequent sections will illustrate how this concept bridges statistical methods with expert information. 


%

\begin{definition}[Informed censored data]
A family of independent random measures $\mathcal{M}=\{\mu_k\}_{k=1,\dots,n}$, each with support $\mathcal{S}_k\subset \mathcal{S}$, will be called informed censored data. The associated measures
\begin{align}\label{def:qmeasures}
d\mathbb{Q}_k:= f_{c} d\mu_k,\quad k=1,\dots,n,
\end{align}
are called the generalized density. Similarly, $\Q:=\Q_1\times\cdots\times \Q_n$ is the generalized likelihood.
\end{definition}

We are interested in deriving properties of carrying out estimation from the generalized likelihood in the presence of informed censored data. Taking logarithms, we use the following convenient notation:
\begin{align}\label{generalized_log_lik}
\ell_n(c\,|\M)=\log\Q(\X_k\in\mathcal{S}_k,\:k=1,\dots,n)
&=\sum_{k=1}^n\log\int_{\mathcal{S}_k} f_{c}(\x) d\mu_k(\x)
\end{align}
for the generalized log-likelihood, emphasizing that we aim at estimating the parameter $c$, given the sample $\M$. As in classical theory, $\ell_n$ will be treated as a random element when deriving asymptotic properties, but in practice, we will only observe a single realization.

Define the following random variable, which plays a key role in the sequel $$W_c=W_c^{(1)}:=-\log\int_{S_1}f_c(x)d\mu_1(x).$$ Then, if $\E[W_c]$ is minimized at the value $\pr$, that is if
$$\pr=\arg\min_{c\in\GG} \E[W_c],$$
then by differentiating under the integral (and assuming $c\mapsto f_c(x)$ is differentiable), we obtain that, for the random variable $$Z_\pr:=-\frac{d}{dc}\log\int_{S_1}f_c(x)d\mu_1(x)\big |_{c=\pr}$$ the following equality holds:
\begin{align}\label{z_zero}\E[Z_\pr]=0.\end{align}
This motivates the following definition:
\begin{definition}[Generalized maximum likelihood estimators]
For the informed censored data $\{\mu_k\}_{k=1,\dots,n}$ and associated variables $W_\pr^{(k)},Z_\pr^{(k)}$, $k=1,\dots,n$, we define
\begin{align*}
&\hat\pr_n=\arg\min_{c\in\GG} \sum_{k=1}^nW^{(k)}_c,\\
& \mbox{and}\;\tilde\pr_n\; \mbox{as the solution to}: \quad \sum_{k=1}^n Z_{\tilde\pr_n}^{(k)}=0,
\end{align*}
whenever such a solution exists in $\GG$. 
\end{definition}
Thus, when it exists (subject to differentiability) and is unique, then $\tilde\pr_n=\hat\pr_n$. Note that the second specification if the finite-sample version of Equation \eqref{z_zero}. 

One does not expect $\hat\pr_n=\arg\min \ell_n(c\,|\M)$ to converge to the real parameter of the underlying distribution. Instead, under certain conditions we will show that it converges to a biased estimate $\pr$, which compromises a purely statistical setup, adding the expert information contained in the $\mu_k$ measures. The resulting asymptotic standard errors will take into account this external information and update the classical formulas accordingly.

Since expert information is not always perfect or statistically coherent, the bias might not be negligible. However, implementing the expert information can lead to considerable variance gains such that the MSE is reduced, see the illustrations later in the paper.


\section{Asymptotic behaviour}\label{sec3}

This section adapts the asymptotic theory found in \cite{van2000asymptotic,kunsch1997mathematische,svg} which are originally tailored for M- and Z- statistics. The main difference in our case is that the random measure $\mu_k$ need not be measurable with respect to $X_k$ (where $(X_k)_{k=1}^n$ is an iid sample from $f$) as it may be a function of $X_k$ and other noise variables, or in some cases not even of $X_k$. However, the proofs nonetheless follow closely, showing that the methodology in the cited references allows for the analysis of more general scenarios without much added complexity.

We begin by giving a definition which guarantees that outside any small neighborhood of the optimal $\pr$, the (population) loss function does not come too close to the optimal value.

\begin{definition}[Well-separated]
We say that the minimizer $\pr\in\GG$ of $\E[W_c]$ is well-separated if for any $\varepsilon>0$,
\begin{align}\label{def:well_sep}
\inf_{c\in\GG,\:||c-\pr||>\varepsilon}\{\E[W_c]\}>\E[W_\pr].
\end{align}
\end{definition}

We may then use the latter condition to prove the consistency of our generalized maximum likelihood estimator.

\begin{theorem}[Consistency]\label{consistency_theo}
Let $\GG$ be a compact set, and assume that the trajectories $c\mapsto W_c(\omega)$ are continuous for all $\omega\in A$ with $\P(A)=1$. Then, subject to the integrability condition $\E\left[\sup_{c\in\GG}|W_c|\right]<\infty,$ we obtain that, as $n\to\infty$,
\begin{align*}
\E[W_{\hat\pr_n}]\to\E[W_{\pr}].
\end{align*}
Furthermore, if $\pr$ is well-separated, then almost surely
\begin{align*}
\hat\pr_n\to\pr.
\end{align*}
\end{theorem}

The intermediate value theorem can help to remove the compactness assumption when we are dealing with one-dimensional parameters since in that case, there is a total order.

\begin{theorem}[One-dimensional consistency]\label{one_dim_theo}
Let $\GG\subset \R$, assume that $c\mapsto Z_c$ is almost surely continuous, $\E[|Z_c|]<\infty$, for any $c\in\GG$, and that there exists $\delta>0$ such that 
\begin{align*}
&\E[Z_c]>0,\quad c\in(\pr,\pr+\delta)\\
&\E[Z_c]<0,\quad c\in(\pr-\delta,\pr).
\end{align*}
Then for $n$ sufficiently large, there is almost surely a consistent solution $\tilde\pr_n$ to $\sum_{k=1}^nZ_{\tilde\pr_n}^{(k)}=0$.
\end{theorem}

In practice, the discrepancy between $\pr$ and the true parameter to estimate is the price to pay for the imputation of expert information. The gain, however, comes from a potential reduction of variance.

For proving asymptotic normality, we introduce the following empirical process, indexed in $\GG$:
\begin{align*}
\nu_n(c)=\sqrt{n}\left( \frac1n\sum_{k=1}^n Z_c^{(k)}-\E[Z_c] \right),\quad c\in\GG.
\end{align*}
\begin{definition}[Asymptotic continuity]
The empirical process $\{\nu(c)\}_{c\in\GG}$ is asymptotically continuous if for all sequences $\{\pr_n\}\subset \GG$ with $\P(||\pr_n-\pr||>\varepsilon)\to0$ as $n\to \infty$, $\forall \varepsilon>0$, we have
$$\P(||\nu_n(\pr_n)-\nu_n(\pr)||>\varepsilon)\to0,\quad \forall \varepsilon>0.$$
\end{definition}
Asymptotic continuity is a well-known assumption that is effective at proving asymptotic normality, but particularly difficult to verify. We first state a general theorem and afterwards proceed to prove a more practical one with some rather restrictive assumptions. It is also worth mentioning that we are by no means stating the least restrictive result possible, but rather those whose proofs can be kept to a decent length.

Recall thaht $p$ is the number of parameters to be estimated. We assume that the $p\times p$ matrix $$M_\pr=\frac{\partial}{\partial c^T}\E[Z_c]\big|_{c=\pr}$$ exists and is of full rank (thus invertible). This is of course an unverifiable assumption, since $f$ is unkown and hence so is $\pr$. However, similarly to classical likelihood theory, we are in practice satisfied if $M_{\hat\pr_n}$ exists and is of full rank.
 
\begin{theorem}[Asymptotic normality via asymptotic continuity]\label{Asymp_norm_1}
Let $\tilde \pr_n$ be a consistent estimator of $\pr$, and let $\nu_n$ be asymptotically continuous at $\pr$. Define
\begin{align}\label{var_expression}
J_\pr=\E\left[ Z_\pr Z_\pr^T \right],\quad V_\pr=M_\pr^{-1}J_\pr [M_\pr^{-1}]^T.
\end{align}
Then
\begin{align*}
\sqrt{n}(\tilde\pr_n-\pr)\stackrel{d}{\to}\mbox{N}(0,V_\pr).
\end{align*}
\end{theorem}

We now provide a theorem that is more restrictive but perhaps easier to verify for practical situations in which the derivatives of $Z_c$ are available.

\begin{theorem}[Asymptotic normality via Lipschitz continuity]\label{asymt_norm_lips}
Let $\tilde \pr_n$ be a consistent estimator of $\pr$, and assume that  in a fixed neighborhood of $\pr$, $c\mapsto Z_c$ is almost surely differentiable with $Z'_c:=\frac{\partial}{\partial c^T}Z_c$ and in matrix norm
\begin{align*}
||Z'_c-Z'_{c_1} ||\le H||c-c_1||
\end{align*}
for an integrable random variable $H$. Then
\begin{align*}
\sqrt{n}(\tilde\pr_n-\pr)\stackrel{d}{\to}\mbox{N}(0,V_\pr),
\end{align*}
with $V_\pr$ as in \eqref{var_expression}.
\end{theorem}

Notice that if the true parameter of the underlying distribution is $\pr_0$, then the above procedure incurs an asymptotic mean square error given by
\begin{align}\label{tradeoff}
\mbox{AMSE}(\widehat{\pr}_n)=\frac1n\, V_\pr+(\pr-\pr_0)^2.
\end{align}
This quantifies the tradeoff between the variance and the bias. If $n$ tends to infinity, the first term goes to 0, but the bias introduced through the informed censoring remains. On the other hand, for $\pr=\pr_0$ the bias term disappears, but in general $V_{\pr_0}>V_{\pr}$. Correspondingly, informed censoring will particularly be useful if one has only a small number $n$ of datapoints available, and one can determine a limiting value $n$ on from which one does not improve the AMSE with respect to the situation without informed censoring (see the following section for concrete computations in particular examples). 

\section{Density randomizations and localized informed censoring}\label{sec:dens_rand}
When an observation is uncertain, the uncertainty around a data point can often be expressed in terms of a probability density (for instance, due to measurement errors, see Section \ref{sec1}). This translates in our setting into $\int d\mu_k=1$. Often, it is a rather reasonable assumption to assume that datapoints are iid observations from \eqref{def:qmeasures} (in the above example for instance in the case that the measuring instruments are not updated once their precision is specified). 

{\color{black} Below we provide two examples where the asymptotic bias and variance can be calculated explicitly, in order to gain insight into how informed censoring fares against benchmark statistical procedures. Another example, using improper measures, will be provided later in Section \ref{sec5}, where the computation of $\int_{S_1}f_c(x)d\mu_1(x)$ is no longer explicit, and numerical integration routines are required (cf.\ Example \ref{tail_index_example}). All numerical results in this and the next sections were carried out in the \texttt{R} language.\footnote{In particular, integration and optimization was done using the \texttt{integrate} and \texttt{optimize} functions in \texttt{R}, respectively.}}

\begin{example}[Normal-normal]\rm
Assume that 
\begin{align*}
f_c(x)=\frac{1}{\sqrt{2\pi}\sigma_1}\exp(-(x-c)^2/(2\sigma_1^2)),\quad c\in\GG=\mathbb{R},
\end{align*}
with $\sigma_1$ known, and that we observe the random datapoints ('expert information')
\begin{align*}
\frac{d\mu_k(x)}{dx}=\frac{1}{\sqrt{2\pi}\sigma}\exp(-(x-X_k-Y_k)^2/(2\sigma^2)),\quad k=1\dots,n,
\end{align*}
with $X_k\sim\mbox{N}(\pr_0,\sigma_1)$ and $Y_k$ being noise random variables with mean $\varepsilon\in\mathbb{R}$.
It follows that 
\begin{align*}
\int_{S_1}f_c(x)d\mu_1(x)=\frac{1}{\sqrt{2\pi(\sigma_1^2+\sigma^2)}}\exp(-\frac12\{\frac{c^2}{\sigma_1^2}+\frac{(X_1+Y_1)^2}{\sigma^2}-(\frac{c}{\sigma_1^2}+\frac{X_1+Y_1}{\sigma^2})^2/(\frac{1}{\sigma_1^2}+\frac{1}{\sigma^2})\}),
\end{align*}
and hence
\begin{align*}
W_c&=\frac{1}{2}\left[\log(2\pi)+\log(\sigma_1^2+\sigma^2)+\frac{c^2}{\sigma_1^2}+\frac{(X_1+Y_1)^2}{\sigma^2}-(\frac{c}{\sigma_1^2}+\frac{X_1+Y_1}{\sigma^2})^2/(\frac{1}{\sigma_1^2}+\frac{1}{\sigma^2})\right],\\
Z_c&=\frac{c}{\sigma_1^2}-\frac{c\sigma^2}{\sigma_1^2(\sigma_1^2+\sigma^2)}-\frac{X_1+Y_1}{\sigma_1^2+\sigma^2}.
\end{align*}

By Theorem \ref{one_dim_theo}, we have consistency, and to the quantity 
\begin{align*}
\pr=\arg\min_{c\in\GG} \E[W_c]=\pr_0+\varepsilon.
\end{align*}
Notice that if the normal expert guesses are symmetric around the true observation, they incur no bias.

Further, we get
\begin{align*}
M_\pr&=\frac{\partial}{\partial c}\E[Z_c]\big|_{c=\pr}=\frac{1}{\sigma_1^2}\left[1-\frac{\sigma^2}{\sigma_1^2+\sigma^2}\right]=\frac{1}{\sigma_1^2+\sigma^2},\\
J_\pr&=(\pr_0+\varepsilon)^2\frac{(1-\frac{\sigma^2}{\sigma_1^2+\sigma^2})^2}{\sigma_1^4}-2(\pr_0+\varepsilon)^2\frac{1-\frac{\sigma^2}{\sigma_1^2+\sigma^2}}{\sigma_1^2(\sigma_1^2+\sigma^2)}+\E[(X_1+Y_1)^2]\frac{1}{(\sigma_1^2+\sigma^2)^2}\\
&=\frac{(\pr_0+\varepsilon)^2}{(\sigma_1^2+\sigma^2)^2}-2\frac{(\pr_0+\varepsilon)^2}{(\sigma_1^2+\sigma^2)^2}+\E[(X_1+Y_1)^2]\frac{1}{(\sigma_1^2+\sigma^2)^2}\\
&=\frac{\mbox{Var}(X_1+Y_1)}{(\sigma_1^2+\sigma^2)^2}
\end{align*}
and finally
\begin{align*}
V_\pr=J_\pr/M_\pr^2=\mbox{Var}(X_1+Y_1).
\end{align*}

In other words, not only does the asymptotic variance not depend on the parameter $\pr$, but it is also independent of the expert variance $\sigma^2$. If $Y_1\equiv0$, this remarkable fact can be interpreted in this setting as follows: if one believes that, on average, the expert guess is exactly centered in the true value, then the spread of the guess is immaterial to the estimation, and in particular we may impute the centerpoints of the guesses ($\sigma=0$) and perform standard maximum likelihood estimation.

The asymptotic mean square error for $\tilde\pr_n$ is given by
$$\mbox{AMSE}(\tilde\pr_n)=\frac1n \mbox{Var}(X_1+Y_1)+\E[Y_1]^2,$$
which shows the bias-variance tradeoff incurred by uncertainty in the centers of the expert guesses. If $(X_1,Y_1)$ is a Gaussian random vector, we get
$$\frac 1n(\sigma_1^2+2\sigma_2^2+\rho\sigma_1\sigma_2)+\epsilon^2,$$
which is optimized for $\varepsilon=0$ and $$\sigma_2^2=\max\{-\rho\sigma_1^2/4,0\}.$$
\end{example}

\begin{example}[Exponential-Gamma]\rm
Assume that 
\begin{align*}
f_c(x)=c\exp(-c x),\quad c\in\GG=\mathbb{R}_+,
\end{align*}
and that we observe the random data points ('expert information')
\begin{align*}
\frac{d\mu_k(x)}{dx}=\frac{\beta_k^{\alpha_k}}{\GG(\alpha_k)}x^{\alpha_k-1}\exp(-\beta_k x),\quad k=1\dots,n,
\end{align*}
with $X_k\sim\mbox{Exp}(\pr_0)$ and
\begin{align*}
\alpha_k=X_k/\sigma^2,\quad \beta_k=1/\sigma^2,\quad \sigma^2>0.
\end{align*}
In other words, the expert guesses are moment-matched gamma densities centered at $X_k$ and with variance $\sigma^2 X_k$, so the uncertainty of the expert information increases with the size of the observation.

It follows that 
\begin{align*}
\int_{S_1}f_c(x)d\mu_1(x)=c\frac{\beta_1^{\alpha_1}}{(\beta_1+c)^{\alpha_1}}=c\left(\frac{1}{1+\sigma^2 c}\right)^{X_1/\sigma^2},
\end{align*}
and hence
\begin{align*}
W_c&=\frac{X_1}{\sigma^2}\log(1+\sigma^2 c)-\log(c),\\
Z_c&=\frac{X_1}{1+\sigma^2 c}-\frac{1}{c}.
\end{align*}

By Theorem \ref{one_dim_theo}, we again have consistency, and to the quantity 
\begin{align*}
\pr=\arg\min_{c\in\GG} \E[W_c]\Leftrightarrow 0=-\frac{1}{c}+\frac{1}{\pr_0}\frac{1}{1+\sigma^2\pr}
\end{align*}
which immediately yields $$\pr=(1/\pr_0-\sigma^2)^{-1},$$
provided that $\sigma^2\pr_0<1$, a necessary assumption. 

Then we may calculate
\begin{align*}
M_\pr&=\frac{\partial}{\partial c}\E[Z_c]\big|_{c=\pr}=-(1/\pr_0-\sigma^2)^2(1+\pr_0\sigma^2),\\
J_\pr&=(1/\pr_0-\sigma^2)^2\E[(X_1\pr_0-1)^2]=(1/\pr_0-\sigma^2)^2.\\
\end{align*}
We collect the terms to obtain
\begin{align*}
V_\pr=J_\pr/M_\pr^2=\frac{1}{(1/\pr_0-\pr_0\sigma^4)^2}.
\end{align*}
For $\sigma=0$ we recover the well-known result for the exponential distribution. In all other cases, however, the variance is larger.

The asymptotic mean square error for $\tilde\pr_n$ is given by
$$\mbox{AMSE}(\tilde\pr_n)=\frac1n \frac{1}{(1/\pr_0-\pr_0\sigma^4)^2}+\frac{\pr_0^2\sigma^4}{(1/\pr_0-\sigma^2)^{2}}.$$
Notice that an oracle expert -- one who makes perfect predictions (i.e., $\sigma^2=0$) -- would establish an estimator $\tilde\pr_n^\circ$ with the asymptotic mean square error
$$\mbox{AMSE}(\tilde\pr_n^\circ)=\frac1n \,\pr_0^2.$$
We may then define the asymptotic efficiency of the expert, relative to the oracle, as
\begin{align*}
e(\tilde\pr_n|\tilde\pr_n^\circ)=\frac{\mbox{AMSE}(\tilde\pr_n)}{\mbox{AMSE}(\tilde\pr_n^\circ)},
\end{align*}
with the efficiency of $1$ being unachievable in practice. Poorer expert information will lead to poorer efficiency, for a given sample size $n$. We may further consider the required expert precision, as a function of a given pre-specified efficiency and sample size. In formulas, for $e\in[1,\infty)$, let $n\ge0$,
\begin{align}\label{effic1}
\sigma(e,n)\quad\mbox{be the solution to}\quad  e(\tilde\pr_n|\tilde\pr_n^\circ)=e.
\end{align}
A plot of such solutions is given in Figure \ref{fig:sigma_surface}, for $\pr_0=1/2$ in the above setup. Concretely, it depicts for every choice $e$ of efficiency and number $n$ of datapoints the limiting value $\sigma$ up to which one still benefits from including the noisy data information. 
\begin{figure}[!htbp]
\centering
\includegraphics[clip, trim=0cm 5cm 0cm 5cm,width=0.8\textwidth]{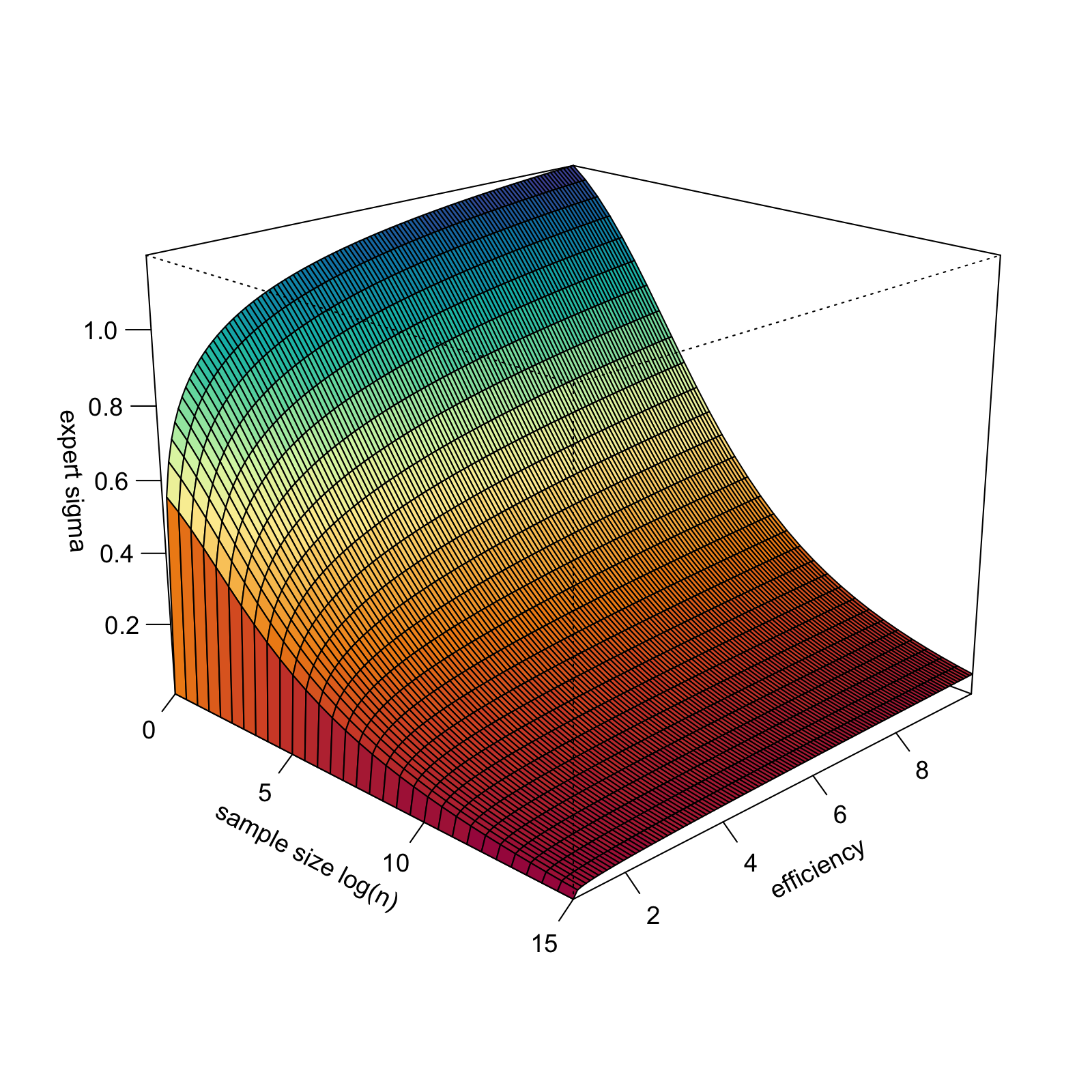}
\caption{Surface plot associated with solutions of Equation \eqref{effic1}, that is, the required expert precision $\sigma$ required to match a certain efficiency $e$ with respect to the oracle ASME, as a function of the sample size $n$.}
\label{fig:sigma_surface}
\end{figure}

An analysis of the required sample size $n$ required by expert information to match the efficiency of an oracle of sample size $n_0\le n$ is now available as a function of $\sigma$. Specifically, let
\begin{align}\label{effic2}
n(n_0,\sigma)\quad\mbox{be the solution to}\quad  e(\tilde\pr_n|\tilde\pr_{n_0}^\circ)=1.
\end{align}
A plot of such solutions is given in Figure \ref{fig:n_surface}, again for $\pr_0=1/2$ as above. As expected, larger $\sigma$ increases the required $n$ for a given $n_0$. Perhaps less obvious is the fact that this effect is exacerbated as $n_0$ grows. One may rationalize that intuitively as the statistical evidence kicking in to reduce the variance, but the expert bias staying constant (regardless of sample size) is ultimately detrimental.

\begin{figure}[!htbp]
\centering
\includegraphics[clip, trim=0cm 5cm 0cm 5cm,width=0.8\textwidth]{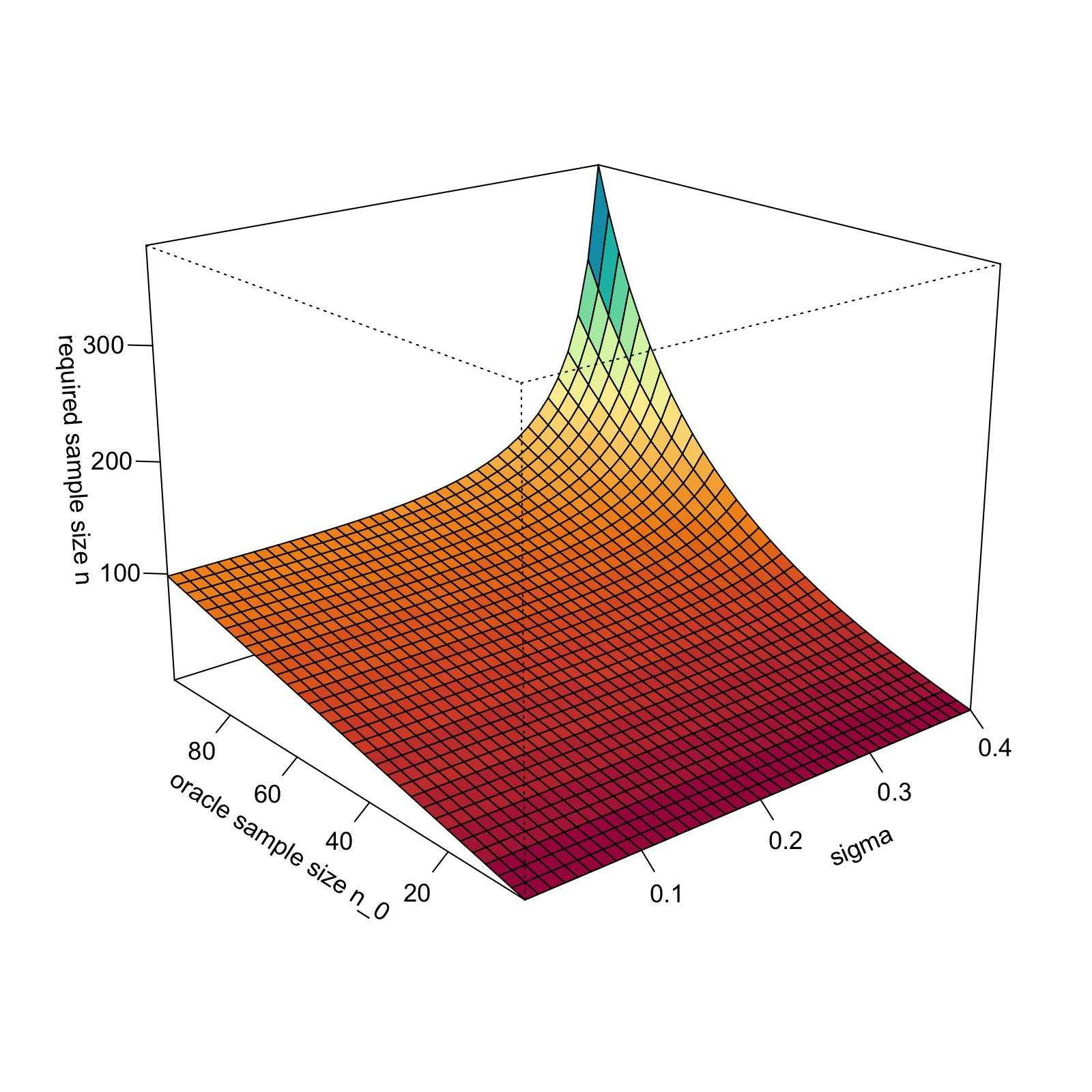}
\caption{Surface plot depicting the solution of Equation \eqref{effic2}: the required expert sample size $n$ to match the ASME of the oracle of sample size $n_0$, as a function of expert precision $\sigma$.}
\label{fig:n_surface}
\end{figure}
\end{example}

{\color{black} 
\begin{remark}[Numerical integration]\rm
Note that the computation of $\int_{S_1}f_c(x)d\mu_1(x)$ is the main ingredient in the above analysis, and also in the optimization procedure. The reader may find similarity with Bayesian computation of normalizing constants, though the present task is fundamentally different. For instance, in Bayesian statistics Markov Chain Monte Carlo (MCMC) methods provide a way to circumvent the calculation of such an integral, whereas no such amendment is possible to compute the perturbed likelihood, so that numerical routines are required in most cases.
\end{remark}
}

\section{Improper randomizations and unbounded informed censoring}\label{sec5}

We introduce the concept of improper randomization, corresponding to the case $\int\mu_k\in (1,\infty]$. This concept loses the interpretation of assigning probabilities towards the uncertainty quantification of datapoints. However, it is a necessary step toward a theory that generalizes right-censoring, a setting very common in biological, financial, and insurance applications, to name a few. 

Specifically, in classical survival analysis, $X$ is a random variable of interest with values in $[0,\infty)$. Often, $X$ is subject to right-censoring, namely, we observe
\begin{align}\label{eq:W_delta_classic}
W = X \wedge C, \quad \delta = \mathds{1}_{\{W = X\}},
\end{align}
where $C$ is another random variable (usually called the censoring mechanism) with values in $[0,\infty]$ describing right-censoring. If $\delta = 1$, then the variable $W$ is said to be \textit{observed}, while if $\delta = 0$, then the variable $W$ is said to be \textit{right-censored}. Letting $(X_k,C_k)_{k=1}^n$ be iid replicates of $(X,C)$, and constructing $W_k$ and $\delta_k$ according to~\eqref{eq:W_delta_classic}, the most common estimation method is to consider the survival log-likelihood
\begin{align*}
\tilde{\ell}_n(\pr)&:=\sum_{k=1}^n\delta_k\log\left(f_\pr(W_k)\right)+(1-\delta_k)\log \left(1-F_\pr(W_k)\right) \\
&=\sum_{k=1}^n\log \int_0^\infty f_\pr(x)\,\delta_i \cdot d\Delta^{W_k}(x) +\log\int_{0}^\infty f_\pr(x) (1-\delta_k)\mathds{1}(W_k\le x) dx,
\end{align*}
where $\Delta^{a}(x)$ denotes the Dirac measure at $a$. In our setting, the above estimation method corresponds to
\begin{align}\label{righ_censoring_mu}
\frac{d\mu_k(x)}{dx}=\delta_k\frac{d\Delta^{W_k}(x)}{dx}+(1-\delta_k)\mathds{1}(W_k\le x)
\end{align}
where $\frac{d\Delta^{W_k}(x)}{dx}$ is the Dirac density at $x$, defined only symbolically.

Observe that in this case, with probability $\P(\delta_1=0),$ the measure $\mu_k$ has a constant  Radon-Nykodym derivative on $[X_1,\infty)$, taking the value $1-\delta_1$. When dealing with informed censoring we now propose two different ways of adapting \eqref{righ_censoring_mu}.

%
%

\begin{example}[Measurement uncertainty]\rm
The simplest adaptation of \eqref{righ_censoring_mu} is to incorporate uncertainty of the measurements $W_k$. Then simply
\begin{align}\label{eq:measurement_uncertainty}
\frac{d\mu_k(x)}{dx}=\delta_kg_{k}{(x)}+(1-\delta_k)G_k(x),
\end{align}
where $G_k$ is a cumulative distribution function, with corresponding density $g_k$ which describes the spread around $W_k$. A less common scenario, which arises in some applications (for instance in life insurance durations), is that the $\delta_k$ are subject to uncertainty and expert information (cf. \cite{Bladt2022}). In that case, binary exogeneous variables $(\eta_k)_{k=1}^n$ -- expert guesses -- are provided, and so we may consider
\begin{align*}
\frac{d\mu_k(x)}{dx}=\eta_kg_{k}{(x)}+(1-\eta_k)G_k(x).
\end{align*}
Notice that, in general, $\delta_k$ and $W_k$ are dependent. A sensible statistical model will assume that, in \eqref{eq:measurement_uncertainty}, $G_k$ is a measurable function of at least $(W_k,\delta_k)$ (and some other noise variables), making their asymptotics untractable in closed form. In that case, bootstrap procedures might be preferable to assess the spread of the resulting estimators.
\end{example}

\begin{example}[Informed censoring]\rm\label{informed_censoring_example}
The adaptation of \eqref{righ_censoring_mu} to incorporate information about $X_i$ rather than $W_i$ is slightly more subtle than for the previous example. Fundamentally, the main difference is the interpretation of the expert information. If the expert information on $X_i$ is perfect, then imputation - or a smooth equivalent - can be handled with the density randomizations of Section \ref{sec:dens_rand}, and it is the obvious choice. However, densities do not degenerate into a positive constant over unbounded intervals, which is the requirement to specify fully censored data without any expert information. 

Consequently, to bridge between fully observed and fully censored datapoints using a parametric expert information measure $\tilde\mu_k(\cdot;\sigma)$, we require that it satisfies the following two conditions:
\begin{enumerate}
\item Degeneracy to the true value. As $\sigma\to0$,
\begin{align*}
\tilde\mu_k(\cdot;\sigma)\to \Delta^{X_k}(\cdot).
\end{align*}
\item Convergence to the improper uniform distribution on $[W_k,\infty)$. As $\sigma\to\infty$,
\begin{align*}
d\tilde\mu_k(\cdot;\sigma)\to \mathds{1}(W_k\le x)dx.
\end{align*}
\end{enumerate}
Putting the pieces together, we can construct the expert information measure as follows:
\begin{align*}\label{}
\mu_k=\delta_k\Delta^{W_k}+(1-\delta_k)\tilde\mu_k.
\end{align*}
\end{example}

{\color{black}
\begin{example}[Tail index estimation]\rm\label{tail_index_example}
An important example motivated from applications where heavy-tailed distributions arise naturally is the case corresponding to the Pareto$(c)$ distribution, with cdf given by
\begin{align}\label{pareto_dist}
F_c(x)=1-(x/x_0)^{-c},\quad x\ge x_0, \:\: c>0,
\end{align}
where $x_0$ is known for the specific situation (or might be simply estimated as the smallest value in the sample). One is usually interested in the tail index $\xi=1/c$, which measures the ''heaviness'' of the tail, with larger values corresponding to heavier tails. In addition to $(W_k,\delta_k)$, assume that we have access to expert guesses $Z_k$ for the size of right-censored observations. One such application coming from the insurance sector is provided in the next section in more detail. Then we may construct an improper measure satisfying the conditions of Example \ref{informed_censoring_example} by specifying the following parametric family of measures (indexed by $\sigma$):
\begin{align}\label{nu_density_example}
\frac{d\tilde\mu_k(x)}{dx}=\min\{\sigma^2,1\} +f_\Gamma(x-W_k+1;\,(Z_k-W_k+1)\sigma^2,\,1/\sigma^2) ,\quad x\ge W_k,
\end{align}
and zero otherwise, where $f_\Gamma(\cdot;\,a,\,b)$ is the Gamma density with shape parameter $a$ and rate parameter $b$. Any other choice satisfying the degeneracy and convergence conditions would work as well. In Figure \ref{nu_density} some of the possible shapes of the above density are depicted, for varying $\sigma$ values. 

\begin{figure}[!htbp]
\centering
\includegraphics[width=0.8\textwidth]{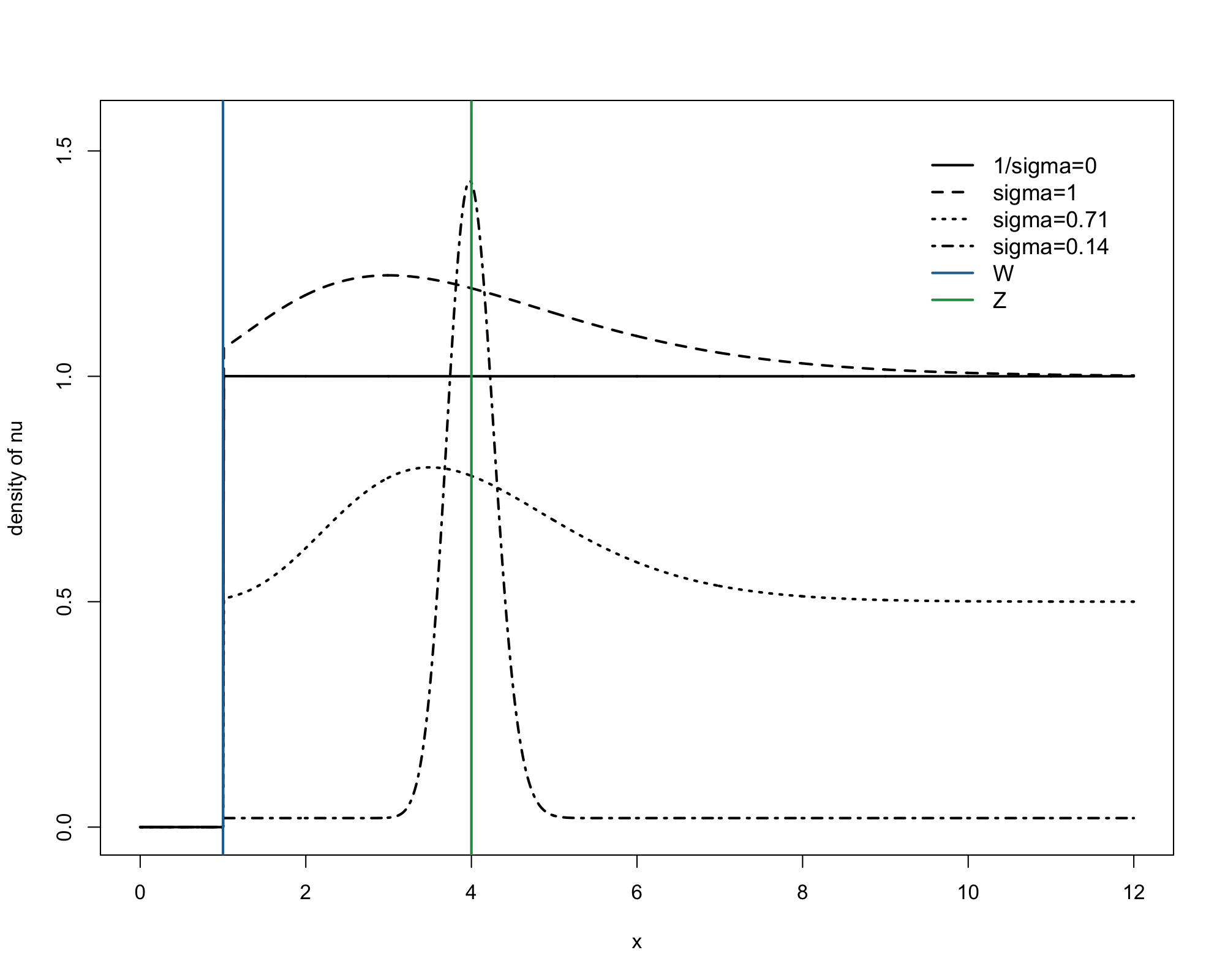}
\caption{Possible shapes for the function in \eqref{nu_density_example}, bridging the two important degenerate cases: the Dirac point mass at $Z$ and the non-informative uniform function on $[W,\infty)$.}
\label{nu_density}
\end{figure}

To illustrate the finite-sample behaviour of our estimator, in comparison to the benchmark cases $\sigma=0$ (imputation of the expert information as true observations) and $\sigma=\infty$ (discarding the expert information and considering a standard survival likelihood), we now perform a short simulation study. We simulate $N=1000$ samples of size $n=200$ with 
\begin{enumerate}
\item $X_k\sim \mbox{Pareto}(2)$,
\item $C_k\sim \mbox{Pareto}(1)$,
\item $Z_k=X_k$ if $\delta_k=1$, and otherwise $Z_k=\max\{X_k,\:X_k+N_k\}$,
\end{enumerate}
where $N_k\sim \mbox{N}(m,s)$, the normal distribution with mean $m$ and standard devation $s$. Thus, the quality of the expert information is modulated by $m$ and $s$. We then numerically explore three scenarios: $$(m,s)=(0,\:1/10),\:(2/10,\:1/10),\:(5/10,\:1/10),$$ corresponding to high, medium, and low expert precision. The results are summarized in Figure \ref{sim_study}. As expected, higher and lower expert precision scenarios steadily and smoothly pull the minimal MSE towards $\sigma=0$ and $\sigma=\infty$, respectively. In each case, there is a sweet spot somewhere in between, which indicates that for finite samples our informed censoring method has the capability of capturing the correct amount of expert information, compared to the benchmarks (located at the asymptotic left and asymptotic endpoints in Figure \ref{sim_study}).

\begin{figure}[!htbp]
\centering
\includegraphics[width=1\textwidth]{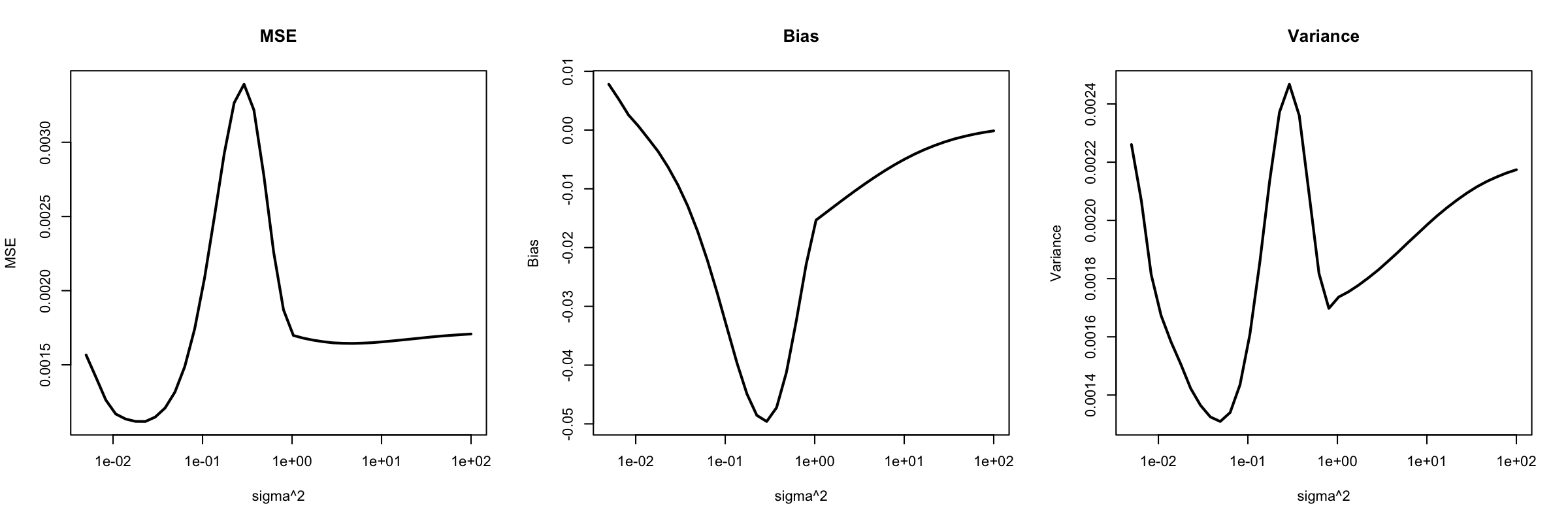}
\includegraphics[width=1\textwidth]{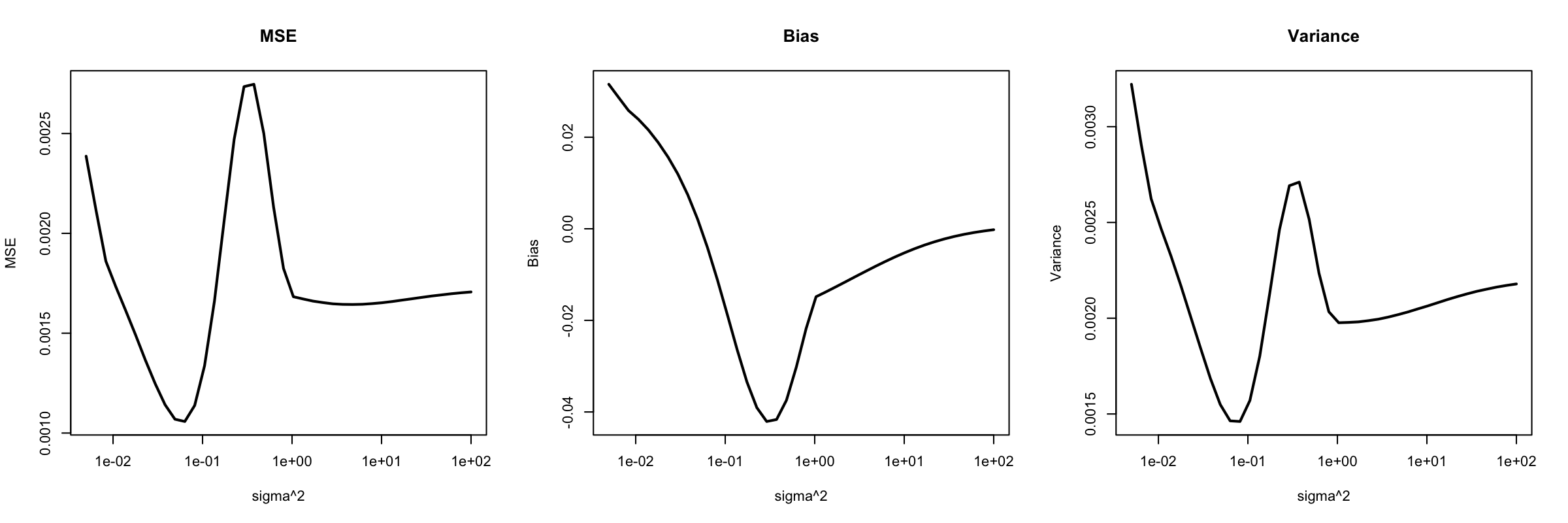}
\includegraphics[width=1\textwidth]{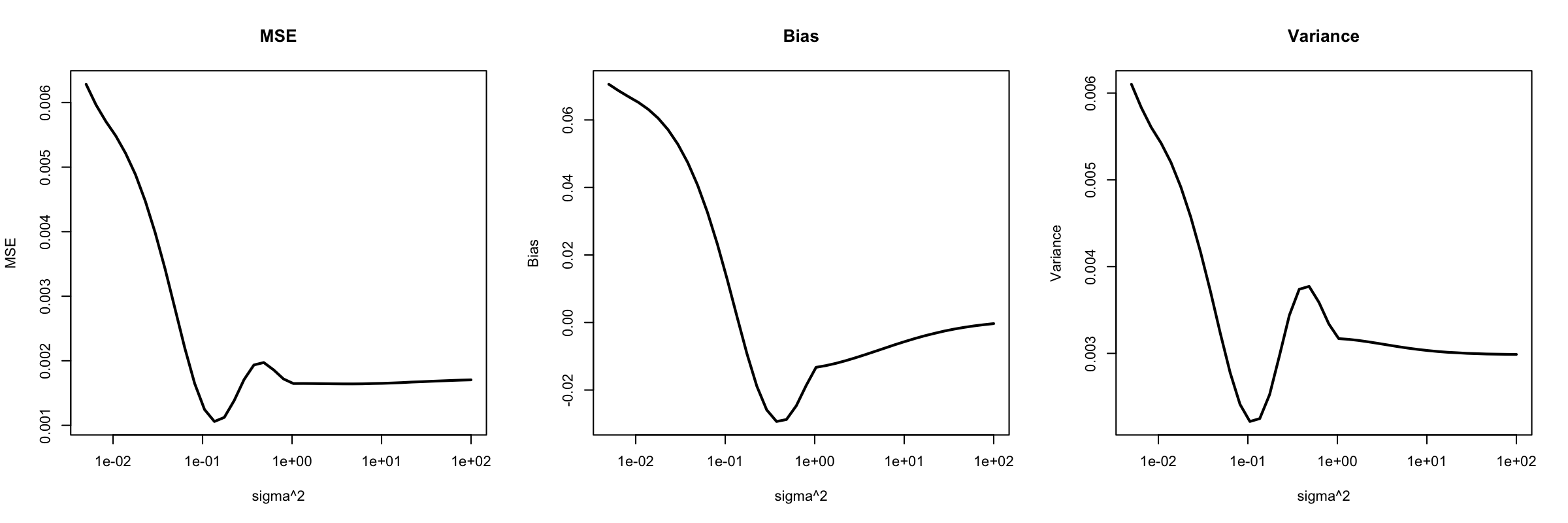}
\caption{{\color{black}Mean square error (MSE), bias and variance of the tail index estimators for varying expert precision. The top, middle and bottom panels correspond to, respectively,  expert information of high, medium, and low precision.}}
\label{sim_study}
\end{figure}
\end{example}
}

\section{Application to real data}\label{sec6}

We consider the Motor Third Party Liability (MTPL) insurance dataset which has previously been studied in  {\cite{abt} and \cite{combined}. The data is private and comes from a direct insurance company operating in the European Union. By 2010, there were $837$ policyholders in the dataset, each with the following three sources of information:
\begin{enumerate}
\item $W_k$: the yearly paid amount paid to the policyholder during the period 1995-2010. Some claims have not been settled yet, and thus are considered to be right-censored.\\
\item $\delta_k$: the censoring indicator, taking the value $1$ if the claim is settled, and $0$ otherwise. Roughly $40.1\%$ of the claims are settled.\\
\item $Z_k$: the \textit{ultimate} value of a claim size, which is the company's internal and undisclosed expert estimation of the eventual claim size for the right-censored cases. For the observed cases, simply $Z_k=W_k$.
\end{enumerate}
We divided $W_k$ and $Z_k$ by $10^6$ for numerical purposes. A careful description of the data can be found in \cite{combined}, where it is also argued that the data has a regularly varying (Pareto-type) tail. 
Presently, we take the latter tail behaviour as granted, and focus on the estimation when incorporating the ultimates information on an individual basis, following the methods of previous sections. Figure \ref{ultimates_vs_paid1} depicts the scatterplot $(W_k, Z_k)$, with the closed claims represented by circles and falling on the diagonal, and the open claims having a vertical segment between their paid amount (again on the diagonal) and their ultimate value (depicted by a triangle).

It is worth noting that the incorporation of the ultimates was studied previously in \cite{combined} without distinguishing expert information between policies, that is, on an aggregate level.

\begin{figure}[!htbp]
\centering
\includegraphics[width=\textwidth]{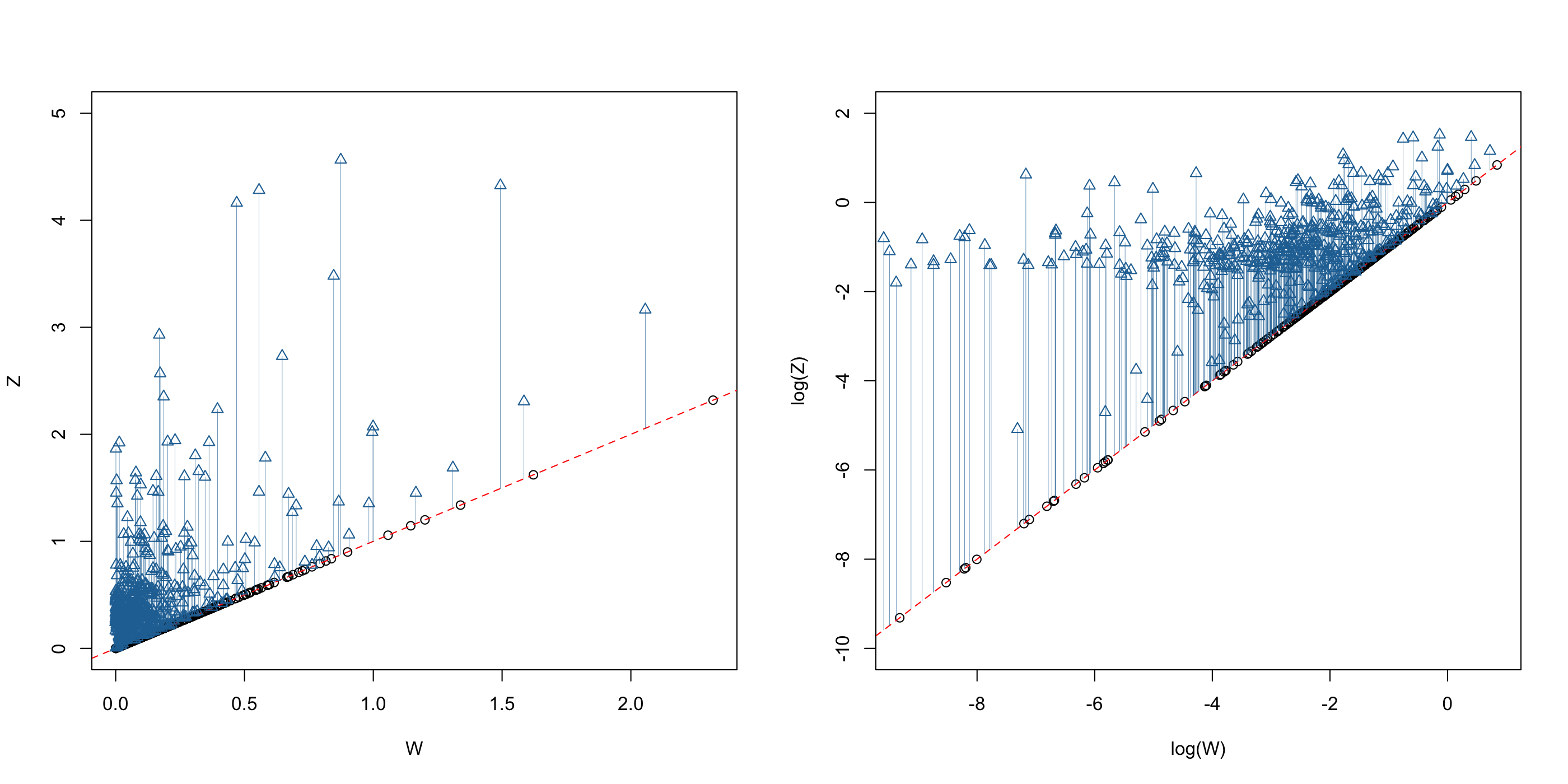}
\caption{Paid amounts $W_k$ versus ultimate values $Z_k$ for the $837$ datapoints in the MTPL insurance dataset. Closed claims are represented by circles and fall directly on the diagonal, while open claims have a vertical segment joining the paid amount (again on the diagonal) and the larger ultimate value, represented by a triangle.}
\label{ultimates_vs_paid1}
\end{figure}

{Since we do not know the quality of the expert information, we follow the approach of Examples \ref{informed_censoring_example} and \ref{tail_index_example}. For this purpose, we have tailored a measure $\tilde\mu_k$, specified by its Radon-Nykodym derivative as \eqref{nu_density_example}}
A priori, one may also want to include the development year of claim $k$ into the specification of $\tilde\mu_k$, but for the present purpose we consider it to be already sufficiently reflected in the values of $Z_k$ and $W_k$.

{\color{black}The distribution for the sample, as mentioned earlier, is chosen as Pareto, and specified by \eqref{pareto_dist}}.
However, this distribution holds only for the tail. Previous analysis (cf.\ \cite{combined}) has shown that the pure Pareto behaviour manifests itself when considering the top $70$ order statistics (or less). Consequently, we take the largest $69$ $W_k$ values and consider them as the new sample, with the scale parameter $x_0$ given by the $70$-th largest $W_k$. We also consider the corresponding concomitants for the $Z_k$ and $\delta_k$. 

We proceed to optimize $\ell_n$ for various values of $\sigma^2$, which incurs in different tail indices {\color{black}$1/c=\gamma=\gamma(\sigma^2)$}.  A curve of the reciprocal quantities (which is usually referred to as the tail index, with larger values indicating heavier tails) is provided in Figure \ref{alphas}. 

We observe something that might seem counterintuitive at first: the informed censored estimator bridges the imputation and full survival methods (and agrees in the extreme cases, by construction), but it does not always lie completely between the two of them. The explanation is that there is a different amount of mass contribution  $\log\int f(x)d\mu_k(x)$ to the loss function for different $\sigma$, which can tilt the estimation towards giving a larger (or smaller) importance  to the corresponding datapoint than in classical likelihood methods. In other words, the height (and not only the shape) of the curves in Figure \ref{nu_density} matters.

The fact that the combined tail index $0.65$ suggested in previous studies (for the same sample fraction) is outside of the imputation and pure survival band is remarkable. However, this has a simple explanation, namely that despite having taken the concomitant with respect to the top order statistics of $W_k$, the ultimates in the new sample may not correspond to the top largest ultimate values. This also highlights the uniqueness and refinedness of incorporating expert information on an individual basis, rather than on the aggregate one.

We also consider the same tail index analysis when the Radon-Nykodym derivative of $\nu_k$ is instead given by {\color{black} the parametric family}
\begin{align}\label{nu_density_example2}
\frac{d\tilde\mu_k(x)}{dx}=\min\{\sigma^2,1\} +f_\Gamma(x;\,Z_k\sigma^2,\,1/\sigma^2) ,\quad x\ge W_k,
\end{align}
which implies an expert variance that changes linearly with the size of the ultimates (in contrast to the differences between ultimates and paid amounts), as for larger ranges the estimate of the ultimates may exhibit more uncertainty. However, both variance specifications lead to very similar tail estimators, implying robustness to misspecification of the parametric expert variance formula (we omit the corresponding plot of this second specification, as it is so similar to the one in Figure \ref{alphas}).

\begin{figure}[!htbp]
\centering
\includegraphics[width=0.8\textwidth]{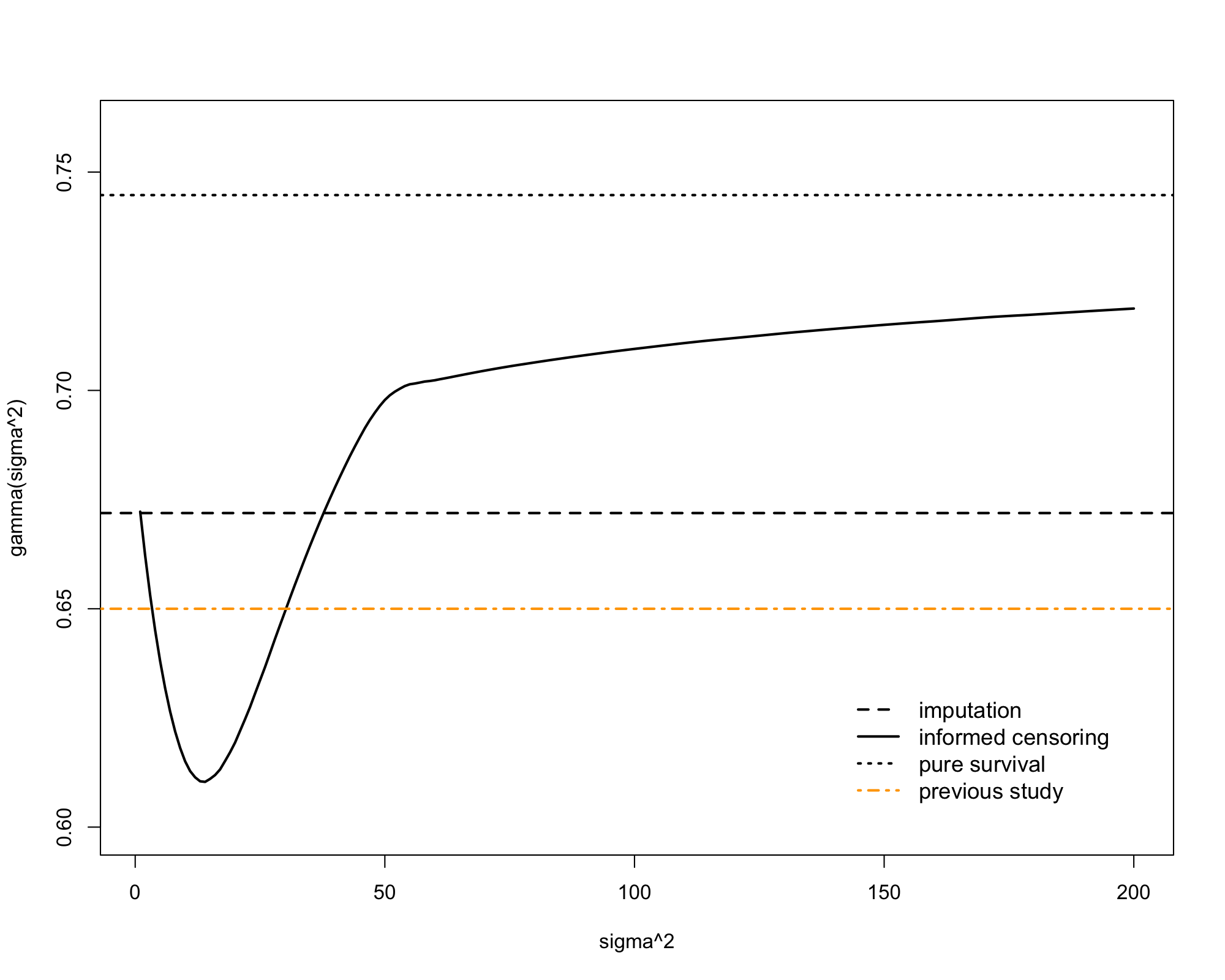}
\caption{The estimated tail index as a function of the expert guess variance $\sigma^2$, together with the value obtained by imputation and by the pure survival approach, as well as the value obtained in a previous study.}
\label{alphas}
\end{figure}

\section{Conclusion and outlook}\label{seccon}
We have provided a statistical framework in which the incorporation of expert information into parametric models is intuitive and hopefully practically useful for both statisticians and practitioners. The setup is general so that it covers varied applications. That said, a very natural context where our setup is useful is the insurance field, where expert judgement is very common.

{\color{black} From a statistical point of view, the} asymptotic theory builds on well-known M-statistics methods of proof, though we have tailor-made the results for easier application to the random measure setup, where the contributions in the loss function may be measurable with respect to any i.i.d. sequence of random elements. Further extensions to the case where the expert information refines with the sample size, and the introduction of covariates are both of great significance and subject of further study. In that context, triangular arrays of row-wise independent random measures will need to be considered.

{\color{black} From a practical perspective, and in particular for mathematically-oriented reserving actuaries, the proposed model provides a new way of dealing with ultimate values in the reserving process. By incorporating these values into a survival analysis setting, and furthermore incorporating the uncertainty of these values through a random measure (which needs to be provided by the practitioner), the accuracy of the projected future losses can increase drastically, as was shown in theoretical and numerical examples. Finally, for an MTPL dataset we have shown how the tail index (a key risk measure for insurance losses) may be chosen in accordance to the degree of expert belief. A related non-parametric approach, which complements the present work, has recently been explored in \cite{Bladt2022}, where the expert information is of a different kind.}


{\color{black} An important shortcoming that needs to be addressed in future work is the automatic choice of the quality of the expert information when only point estimates are available, that is, when the expert cannot provide a random measure for some/all points. Concretely, when a random measure is not available, but only summarizing statistics from it, say a quantile or its mean, then
the inverse problem -- determining expert precision from the best predictive model (say, using cross-validation) -- is an interesting dual specification of our model that may merit closer investigation.}

While the storyline of the illustration focuses on the actuarial application, we would like to reiterate that the proposed framework should be useful in other application areas as well, and it will be interesting in future research to ponder that potential, and possibly sharpen some of the results in the direction of the needs of the respective particular application.

\appendix
\section{Proofs}\label{proofs}

\begin{proof}[Proof of Theorem \ref{consistency_theo}]
Let $\delta>0$ and $c\in\GG$, and define
$$w(\delta,c)=\sup_{c_1\in\GG,\:||c-c_1||<\delta}\{|W_c-W_{c_1}|\}.$$
By continuity we have that $w(\delta,c)\to0$ almost surely, as $\delta\downarrow0$. Continuous functions on compact intervals are bounded, so by dominated convergence we obtain $$ \lim_{\delta\to0}\E[w(\delta,c)]\to0.$$
In other words, given $\varepsilon>0,$ there exists $\delta(c)>0$ such that $\E[w(\delta(c),c)]<\varepsilon/2.$ We define $B_c=\{c_1\in \R^p:||c-c_1||<\delta(c)\},$ so that $\{B_c\}_{c\in \GG}$ is a covering of $\GG$ with open sets. Since $\GG$ is compact, there esists a finite sub-covering $B_{c_1}, \dots,B_{c_m}$, and for each $c\in B_{c_j}$ we get $$|W_c-W_{c_j}|\le w(\delta(c_j),c_j).$$ 
Then, almost surely,
\begin{align*}
\sup_{c\in\GG}\left| \frac 1n \sum_{k=1}^nW_c^{(k)}- \E[W_c]\right|&\le\max_{1\le j\le m}\left| \frac 1n \sum_{k=1}^nW_{c_j}^{(k)}- \E[W_{c_j}]\right|\\
&\quad +\max_{1\le j\le m}\left\{ \frac 1n \sum_{k=1}^nw^{(k)}(\delta(c_j),c_j)\right\}\\
&\quad +\max_{1\le j\le m}\left\{ \E[w(\delta(c_j),c_j)]\right\},
\end{align*}
and the latter converges to $2\max_{1\le j\le m}\left\{ \E[w(\delta(c_j),c_j)]\right\}<\varepsilon$.
With this convergence we can then assert 
\begin{align*}
0\le\E[W_{\hat\pr_n}-W_\pr]&= \frac 1n \sum_{k=1}^n(W_{\hat\pr_n}^{(k)}-W_{\pr}^{(k)})-\frac 1n \sum_{k=1}^n(W_{\hat\pr_n}^{(k)}-W_{\pr}^{(k)})+\E[W_{\hat\pr_n}-W_{\pr}]\\
&\le -\frac 1n \sum_{k=1}^n(W_{\hat\pr_n}^{(k)}-W_{\pr}^{(k)})+\E[W_{\hat\pr_n}-W_{\pr}]\\
&\le\left|\frac 1n \sum_{k=1}^n(W_{\hat\pr_n}^{(k)}-\E[W_{\hat\pr_n}])\right|+\left|\frac 1n \sum_{k=1}^n(W_{\pr}^{(k)}-\E[W_{\pr}])\right|\\
&\le 2\sup_{c\in\GG}\left| \frac 1n \sum_{k=1}^nW_c^{(k)}- \E[W_c]\right|\to0,
\end{align*}
as $n\to \infty$, which proves the first claim. 

Now, assume $\pr$ is well-separated. Then if $\hat\pr_n$ does not converge to $\pr$, there exists $\varepsilon_0>0$ such that for infinitely many $n$, $||\hat\pr_n-\pr||>\varepsilon_0$, which by \eqref{def:well_sep} implies that $\E[W_{\hat\pr_n}]$ does not converge to $\E[W_\pr],$ a contradiction. Thus, $\hat\pr_n\to\pr$, establishing the second claim.
\end{proof}

\begin{proof}[Proof of Theorem \ref{one_dim_theo}]
Given $0<\varepsilon<\delta$, the the strong law of large numbers, there exists almost surely an $n$ such that 
\begin{align*}
\frac1n \sum_{k=1}^nZ_{\pr-\varepsilon}^{(k)} <0< \frac1n \sum_{k=1}^nZ_{\pr+\varepsilon}^{(k)},
\end{align*}
which by continuity implies that that there exists a $\tilde\pr_n$ with $ \sum_{k=1}^nZ_{\tilde\pr_n}^{(k)}=0$ and $|\tilde\pr_n-\pr|<\varepsilon$.
\end{proof}

\begin{proof}[Proof of Theorem \ref{Asymp_norm_1}]
By definition we have that
\begin{align*}
\sum_{k=1}^nZ_{\tilde\pr_n}^{(k)}=0,\quad\E[Z_\pr]=0.
\end{align*}
Note now that by differentiability,
\begin{align*}
\E[Z_{\tilde\pr_n}-Z_\pr]=M_\pr(\tilde\pr_n-\pr)+A_{n,\pr},
\end{align*}
with $A_{n,\pr}/||\tilde\pr_n-\pr||\to0$ in probability. Now, by asymptotic continuity,
\begin{align*}
\frac1n \sum_{k=1}^n(Z_{\tilde\pr_n}^{(k)}-\E[Z_{\tilde\pr_n}])=\frac{\nu_n(\tilde\pr_n)}{\sqrt{n}}&=\frac{\nu_n(\pr)}{\sqrt{n}}+B_{n,\pr}\\
&=\frac1n \sum_{k=1}^nZ_{\pr}^{(k)}-\E[Z_{\pr}]+B_{n,\pr},
\end{align*}
with $\sqrt{n}B_{n,\pr}\to 0$ in probability.

Consequently, and since $\E[Z_\pr]=0$.
\begin{align*}
0&=\E[Z_{\tilde\pr_n}-Z_\pr]+\frac1n \sum_{k=1}^n(Z_{\tilde\pr_n}^{(k)}-\E[Z_{\tilde\pr_n}])\\
&=M_\pr(\tilde\pr_n-\pr)+\frac1n \sum_{k=1}^nZ_{\pr}^{(k)}+A_{n,\pr}+B_{n,\pr}.
\end{align*}
By the central limit theorem, $\nu_n(\pr)$ converges in distribution to a Normal random variable. The above equality then implies that necessarily $\sqrt{n}||\tilde\pr_n-\pr||$ is uniformly tight. Thus
$$0=\sqrt{n}M_\pr(\tilde\pr_n-\pr)+\nu_n(\pr)+C_{\pr,n},$$
with $C_{\pr,n}\to0$ in probability. Then by Slutsky's theorem we finally obtain that $\sqrt{n}(\tilde\pr_n-\pr)$ has the same asymptotic limit as $-M_\pr^{-1}\nu_n(\pr)$.

To identify the parameters of the distribution, we note that by the central limit theorem, $-\nu_n(\pr)$, is asymptotically $\mbox{N}(0,J_\pr)$, and so applying the Delta method yields 
$$\sqrt{n}(\tilde\pr_n-\pr)\stackrel{d}{\to}\mbox{N}(0,M_\pr^{-1}J_\pr [M_\pr^{-1}]^T),$$
as desired.
\end{proof}

\begin{proof}[Proof of Theorem \ref{asymt_norm_lips}]
By the mean value theorem, 
\begin{align*}
0=\frac{1}{n}\sum_{k=1}^nZ_{\tilde\pr_n}^{(k)}=\frac{1}{n}\sum_{k=1}^nZ_{\pr}^{(k)}+(\tilde\pr_n-\pr)\frac{1}{n}\sum_{k=1}^n{Z'_{G}}^{(k)},
\end{align*}
for a random variable $G$ with $$||G-\pr||\le||\tilde\pr_n-\pr||.$$ Then writing
\begin{align*}
0=\frac{1}{n}\sum_{k=1}^nZ_{\pr}^{(k)}+(\tilde\pr_n-\pr)\frac{1}{n}\sum_{k=1}^n{Z'_{\pr}}^{(k)}+
(\tilde\pr_n-\pr)\frac{1}{n}\sum_{k=1}^n({Z'_{G}}^{(k)}-{Z'_{\pr}}^{(k)}),
\end{align*}
we obtain
\begin{align*}
\left|\frac{1}{n}\sum_{k=1}^nZ_{\pr}^{(k)}+(\tilde\pr_n-\pr)\frac{1}{n}\sum_{k=1}^n{Z'_{\pr}}^{(k)}\right|&\le
||\tilde\pr_n-\pr||\cdot||G-\pr||\frac{1}{n}\sum_{k=1}^nH^{(k)},\\
&\le A_{n}||\tilde\pr_n-\pr||^2
\end{align*}
where $A_n=\frac{1}{n}\sum_{k=1}^nH^{(k)}$ is bounded in probability, since it converges to a constant by the strong law of large numbers. In the same way,
\begin{align*}
\frac{1}{n}\sum_{k=1}^n{Z'_{\pr}}^{(k)}=\E[Z'_\pr]+B_n=M_\pr+B_n,
\end{align*}
where $B_n$ converges to zero in probability.
Consequently,
\begin{align*}
\frac{1}{||\tilde\pr_n-\pr||^2}\left|\frac{1}{n}\sum_{k=1}^nZ_{\pr}^{(k)}+(\tilde\pr_n-\pr)[M_\pr+B_n]\right|
\end{align*}
is bounded in probability. Now since $\frac{1}{\sqrt{n}}\sum_{k=1}^nZ_{\pr}^{(k)}$ is also bounded in probability, the same is true for the first-term in its Taylor expansion: $\sqrt{n}(\tilde\pr_n-\pr)$. Thus
\begin{align*}
n\left|\frac{1}{n}\sum_{k=1}^nZ_{\pr}^{(k)}+(\tilde\pr_n-\pr)[M_\pr+B_n]\right|=n\left|\frac{1}{n}\sum_{k=1}^nZ_{\pr}^{(k)}+(\tilde\pr_n-\pr)M_\pr+C_\pr\right|
\end{align*}
is uniformly tight, where $\sqrt{n}C_\pr$ converges to zero in probability. This implies that 
$$\frac{1}{n}\sum_{k=1}^nZ_{\pr}^{(k)}=-(\tilde\pr_n-\pr)M_\pr+\tilde{C}_\pr,$$
with $\tilde{C}_\pr\sim C_\pr.$ The rest of the proof is then identical to the end of the proof of Theorem \ref{Asymp_norm_1}.
\end{proof}

\textbf{Acknowledgement.} 
The authors would like to acknowledge financial support from the Swiss National Science Foundation Project 200021\_191984.

\textbf{Declaration.} The authors declare no conflict of interest related to the current manuscript.

\bibliographystyle{apalike}
\bibliography{informed_censoring.bib}

%
%
%
%
\end{document}